\documentclass[runningheads]{llncs}

\usepackage[comma,authoryear,round,sectionbib]{natbib}

\usepackage{graphicx}
\usepackage{wrapfig}
\usepackage{booktabs}
\usepackage{amsmath,amsfonts,amssymb}
\usepackage{stmaryrd}
\usepackage[boxed]{algorithm}
\usepackage[noend]{algorithmic}

\usepackage{turnstile}
\usepackage{mathrsfs}
\usepackage{enumitem}
\usepackage{amsthm}
\usepackage[super]{nth}
\usepackage[toc,page]{appendix}
\usepackage{listings}
\usepackage{mathtools}
\usepackage{csquotes}
\usepackage{multicol}
\usepackage{fancyvrb}
\usepackage{makecell}
\usepackage{mdframed}
\usepackage[toc,page]{appendix}
\usepackage{mathdots}
\usepackage{empheq}
\usepackage{appendix}
\usepackage{bm}
\usepackage{pgfplots}
\usepackage{enumitem}
\usepackage{float}

\usepackage{nicefrac}

\usepackage{array}
\usepackage{blkarray}
\newcommand{\matindex}[1]{\mbox{\scriptsize#1~~}}

\usepackage{tikz}
\usetikzlibrary{automata, positioning, calc, shapes, arrows, fit}

\usepackage[unicode=true, bookmarks=false, breaklinks=true, pdfborder={0 0 1}, colorlinks=true, citecolor=blue]{hyperref}

\usepackage{pifont}
\definecolor{MyGreen}{rgb}{0, 0.7, 0}
\definecolor{MyRed}{rgb}{0.8, 0, 0}

\usepackage{caption}
\usepackage{boxedminipage}
\usepackage{xspace}

	\newtheorem{fact}{Fact}%

\newcommand{\midd}{\mathbin{:}}

\title{Random Assignment Under Bi-Valued Utilities:\\ Analyzing Hylland-Zeckhauser, Nash-Bargaining, and other Rules}

\author{Haris Aziz\inst{1,}\inst{2}  \and Ethan Brown\inst{1}}

\institute{UNSW Sydney, Australia  \\
\email{\{haris.aziz,e.brown\}@unsw.edu.au}
\and
Data61 CSIRO\\
\medskip
}

\makeatletter
\newcommand{\@chapapp}{\relax}%
\makeatother

\begin{document}
	\maketitle

		\begin{abstract}
		The Hylland-Zeckhauser (HZ) rule is a well-known rule for random assignment of items. The complexity of the rule has received renewed interest recently with Vazirani and Yannakakis (2020) proposing a strongly polynomial-time algorithm for the rule under bi-valued utilities, and making several general insights. We study the rule under the case of agents having bi-valued utilities. We point out several characterizations of the HZ rule, drawing clearer relations with several well-known rules in the literature. As a consequence, we point out alternative strongly polynomial-time algorithms for the HZ solution.
		We also give reductions from computing the HZ solution to computing well-known solutions based on leximin or Nash social welfare. An interesting contrast is that the HZ rule is group-strategyproof  whereas the unconstrained competitive equilibrium with equal incomes rule is not even strategyproof. We clarify which results change when moving from 1-0 utilities to the more general bi-valued utilities.  
		Finally, we prove that the closely related Nash bargaining solution violates envy-freeness and strategyproofness even under 1-0 utilities. 
		\end{abstract}

	
		\section{Introduction}

		%


		Competitive Equilibrium with Equal Incomes (CEEI) is one of the most fundamental solution concepts in resource allocation~\citep{Budi11a,Moul03a,Vari74a}. The concept is based on the idea of a market-based equilibrium which underpins classical economics and has been referred to as the crown jewel of mathematical economics.
		In CEEI, each agent is considered to have equal budget of unit one to spend. An assignment of items satisfies CEEI if, for some price vector for the items, the supply meets demand.
		In other words, the agents get allocations that give them the maximum possible utility.
		CEEI is a well-established in economics because it is based on the idea of a market equilibrium. It is an attractive solution concept because it implies envy-freeness and Pareto optimality. 

		We consider the problem of allocating $n$ items among $n$ agents. Agents have additive cardinal utilities over the items.
		If we view the items as divisible and do not impose any limits on the amount of items given to agents, then CEEI is well-understood. Under additive utilities, CEEI is equivalent to maximizing Nash social welfare. The equivalence follows from an analysis of the Eisenberg-Gale convex program that maximizes the Nash welfare within the assignments of divisible items. However, such a characterization disappears when each agent has a demand for exactly one unit of items. The constraint of unit capacities is especially critical when the fractions of items given to agents are interpreted as probabilities and the goal is to probabilistically find an assignment in which each agent gets one item. If each agent gets one unit of items, then the given probabilistic assignment can be instantiated into a lottery over perfect matchings using Birkhoff's theorem. 

		In this paper, we focus on the pseudo-market rule proposed by \citet{HyZe79a}  that is inspired by CEEI. We will refer to the rule as the HZ rule. HZ can be viewed as the suitable 
		CEEI solution for \emph{probabilistic} or \emph{random} assignment of \emph{indivisible} items.\footnote{There are several works on probabilistic assignment of items to agents~\citep{BoMo01a,ABS13a,ABBM14a,KMN11a,RKFZ14a,KaSe06a}.} The HZ rule has been referred to as CEEI in the literature. We will not use the term CEEI for HZ so that it is clear that we assume the unit-demand requirement when referrering to the HZ solution. The complexity of computing the HZ solution has been open for 40 years\citep{VaYa20a,Seth10a}. 
		Recently, \citet{VaYa20a} explored the computational complexity of the HZ rule. They make several general insights including the fact that the HZ solution can be irrational. They present a strongly polynomial-time algorithm to compute the HZ solution under bi-valued utilities. 
We also focus on the random assignment problem under (cardinal) bi-valued utilities. Bi-valued utilities are meaningful in any scenario in which agents partition items among better and worse and we give them flexibility to express the intensity of preferences. Cognitively, an agent may have a threshold and may prefer items that reach that threshold level. This threshold level may indicate usability or acceptable quality level. Bi-valued are also more general than binary (1-0) utilities that are studied in resource allocation~\citep{BKV18a,HPPS20a}, random assignment~\citep{BoMo04a}, approval voting~\citep{BrFi07c}, probabilistic voting~\citep{BMS05a} and committee voting~\citep{ABC+16a}.

		\paragraph{Contributions}
		We consider random assignment problem under (cardinal) bi-valued utilities.
		We prove several characterizations of the HZ rule, drawing clearer relations with several well-known rules in the literature. As a consequence, we point out alternative strongly polynomial-time algorithms for the HZ solution. In particular, we show that 
		the Extended Probabilistic Serial (EPS) by \citet{KaSe06a} (which is designed for egalitarian objectives) also returns the HZ solution. 
		For bi-valued utilities, we also provide a reduction from computing the HZ rule to computing the leximin or maximum Nash welfare solution. Another structural insight we have is the following one:  for all HZ solutions under bi-valued utilities, each item gets the same price in all the solutions as long as the underlying dichotomous preferences do not change. 
		We also show the following interesting contrast. Under bi-valued utilities, the HZ rule is group-strategyproof whereas the (unconstrained) CEEI rule is not even strategyproof. Therefore, an innocuous-looking relaxation of the unit-demand requirement leads to completely different strategic properties. Our result regarding the manipulability of the CEEI rule also implies that for bi-valued utilities, the MNW rule of \citet{CKM+16a} for indivisible goods is not strategyproof.
		Finally, we study the Nash Bargaining rule that has close links with MNW, CEEI, and HZ rules. We prove that the Nash bargaining solution violates envy-freeness and strategyproofness even under 1-0 utilities.

		\section{Preliminaries}
		%

		An assignment problem is a triple $(N,O,u)$ such that $N=\{1,\ldots, n\}$ is the set of agents, $O=\{o_1,\ldots, o_n\}$ is the set of items, and $u=(u_1,\ldots, u_n)$ is the utility profile which specifies for each agent $i\in N$ utility function $u_{i}$ where $u_{ij}$ denotes the utility of agent $i$ for item $o_j$. We assume that for each $j\in \{1,\ldots, n\}$,  $u_{ij}>0$ for some $i\in N$ and for each $i\in N$,  $u_{ij}>0$ for some $j\in \{1,\ldots,n\}$. 


		A fractional assignment $x$ is a $(n\times n)$ matrix $[x_{ij}]$ such that $ x_{ij} \in [0,1]$ for all $i\in N$, and $o_j\in O$, and  $\sum_{i\in N}x_{ij}= 1$ for all $o_j\in O$.
		The value $x_{ij}$ represents the fraction of item $o_j$ being allocated to  agent $i$. In a probabilistic context, the value $x_{ij}$ has a natural interpretation as the probability of $i$ getting item $o_j$.
	Each row $x_i=(x_{i1},\ldots, x_{in})$ represents the \emph{allocation} of agent $i$. 
		We will denote the set of all allocations by $\mathcal{A}$.
		An allocation $x_i$ is balanced if $\sum_{o_j\in O}x_{ij}=1$.
		We will denote the set of all balanced allocations by $\mathcal{A}_b$.
		For any allocation $x_i$, we will refer to $\sum_{o_j\in O}x_{ij}$ as the \emph{size} of the allocation. 
		The set of columns correspond to the items $o_1,\ldots, o_n$.
		A fractional assignment is \emph{discrete} if $x_{ij}\in \{0,1\}$ for all $i\in N$ and $o_j\in O$. 
		Let the set of all fractional assignments be $\mathcal{F}$.
		A fractional assignment is \emph{balanced} if $\sum_{o_j\in O}x_{ij}=1$ for all $i\in N$. 
		Let the set of all balanced fractional assignments be $\mathcal{F}_b$.


		The expected utility received by agent $i$ from assignment $x$ is $u_i(x_i)=\sum_{o_j\in O}x_{ij}u_{ij}.$ 

		We say that the utility functions are binary or 1-0 if $u_{ij}\in \{0,1\}$ for all $i,j\in N$. We say that the utility functions are bi-valued if for all $i$, $u_{ij}\in \{\alpha_i,\beta_i\}$ where $\alpha_i>\beta_i\geq 0$. Both binary and bi-valued preferences are forms of dichotomous preferences. We will denote by $D_i$ the set of items most preferred by agent $i$. 	For any bi-valued utility function involving values $\alpha_i, \beta_i$, we call  by \emph{binary-reduced} those utility functions in which $\alpha_i$ is turned into 1 and $\beta_i$ is turned into zero.

		An assignment $x$ is \emph{Pareto optimal (PO)} if there exists no other assignment $y$ such that $u_i(y_i)\geq u_i(x_i)$ for all $i\in N$ and $u_i(y_i)> u_i(x_i)$ for some $i\in N$.
		An assignment $x$ is \emph{Pareto optimal among balanced assignments} if there exists no balanced assignment $y$ such that $u_i(y_i)\geq u_i(x_i)$ for all $i\in N$ and $u_i(y_i)> u_i(x_i)$ for some $i\in N$.

		We present a simple routine to turn an unbalanced assignment into a balanced one. We will refer to it as the balancing operation.
					\paragraph{Balancing operation} 
		If an assignment is not balanced, we consider sets 
		$N^+= \{i\in N\midd \sum_{o_j\in O}x_{ij}>1\}$ and 
		$N^-= \{i\in N\midd \sum_{o_j\in O}x_{ij}<1\}$. 
		Each agent  $i\in N^+$ gives away the least preferred items from her allocation so as to ensure that her allocation $x_i$ has size $1$. The donated items are then given to $N^-$ arbitrarily to ensure that $\sum_{o_j\in O}x_{ij}=1$ for all $i\in N$.



		\section{Solution Concepts}

		We present a few prominent solution concepts starting with the HZ solution.

		An assignment $x$ is an \emph{HZ solution} if $x\in \mathcal{F}_b$ and  there exists a price vector $p=(p_1,\ldots,p_n)$ that specifies the price $p_j$ of item $o_j$ such that the maximal share that each $i\in N$ can get with
		 budget $1$ is $$x_i\in \{x_i'\in \mathcal{A}_b \midd x_i'\in \text{argmax}\{u_i(x_i'):\sum_{o_j\in O}x_{ij}'\cdot(p_j)\leq 1\}\}$$ and {
		 $$ \sum_{o_j \in O}{p_j x_{ij}}= \min \{\sum_{o_j \in O}{p_j y_{ij}} \;|\; y_i \in \mathcal{A}_b, \sum_{o_j \in O}{u_{ij} y_{ij}} \geq \sum_{o_j \in O}{u_{ij} x_{ij}}\}.
	$$}
	 
		 We will refer to the rule that returns the HZ solution as the \emph{HZ rule}. 

		 A closely related concept is CEEI.
		An assignment $x$ satisfies \emph{competitve equilibrium with equal incomes (CEEI)} if there exists a price vector $p=(p_1,\ldots,p_n)$ that specifies the price $p_j$ of item $o_j$ such that the maximal share that each $i\in N$ can get with
		 budget $1$ is $$x_i\in \{x_i'\in \mathcal{A}\midd x_i'\in \text{argmax}\{u_i(x_i'):\sum_{o_j\in O}x_{ij}'\cdot(p_j)\leq 1\}\}$$ and {
		 $$ \sum_{o_j \in O}{p_j x_{ij}}= \min \{\sum_{o_j \in O}{p_j y_{ij}} \;|\; y \in \mathcal{A}, \sum_{o_j \in O}{u_{ij} y_{ij}} \geq \sum_{o_j \in O}{u_{ij} x_{ij}}\}.
	$$}
		The CEEI rule returns a CEEI assignment.  
		 CEEI~\citep{Vari74a}  coincides with the market equilibrium notion studied in \citep{Vazi07a}. Note that the HZ solution can be viewed as CEEI with the additional constraint that each agent gets one unit of items. 
		 {For both CEEI and HZ, we require that if an agent can get multiple maximal shares under given prices, she gets a cheapest possible maximal share.} 

		The MNW (Maximum Nash Welfare) rule returns
		an assignment $x$ that maximizes the Nash social welfare:
		$$x\in \arg\max_{x'\in \mathcal{F}} \prod_{i\in N}u_i(x_i').$$

		For two vectors $\vec{u}, \vec{v} \in \mathbb{R}^k$, we say that $\vec{u}$ leximin-dominates $\vec{v}$, written $\vec{u} \succ_{lex} \vec{v}$, if there exists an $i \leq k$ such that $\vec{u}_j = \vec{v}_j,$ for all $j < i$, and $\vec{u}_i > \vec{v}_i$. Finally, $\pi$ is leximin optimal if there is no $\pi'$ such that $\vec{u}(\pi') \succ_{lex} \vec{u}(\pi)$. 	The \emph{leximin rule} is the rule that returns a leximin optimal assignment. 

		CEEI, MNW, and leximin may not return a balanced assignment.

		We will say that two rules are \emph{equivalent} if they result in the same utilities for the agents.

		\begin{example}\label{example:main}
			Consider the following instance with two agents and items.
			\begin{center}
					\setlength{\tabcolsep}{6pt}
					\begin{tabular}{ccccccccc}
						& $o_1$ & $o_2$\\
						\midrule
						$1$ &$3$&$2$&\\
						$2$ &$1$&$0$\\
					\end{tabular}
				\end{center}
	
				For this instance, the HZ, MNW constrained to balanced assignments, and the MNW solution are as follows. 
	
						HZ solution $=$
				 \begin{blockarray}{ccccccccccc}

				 		&&\matindex{$o_1$}&\matindex{$o_2$}&\\
				 	    \begin{block}{l(cccccccccc)}
				 			\matindex{$1$}~~& &$\nicefrac{1}{2}$&$\nicefrac{1}{2}$ \\
				 			\matindex{$2$}& &$\nicefrac{1}{2}$&$\nicefrac{1}{2}$ \\
				 	    \end{block}
				 	  \end{blockarray}

	
						MNW solution constrained to balanced assignments $=$
				 \begin{blockarray}{ccccccccccc}

				 		&&\matindex{$o_1$}&\matindex{$o_2$}&\\
				 	    \begin{block}{c(cccccccccc)}
				 			\matindex{$1$}& &$0$&$1$ \\
				 			\matindex{$2$}& &$1$&$0$ \\
				 	    \end{block}
				 	  \end{blockarray}

						CEEI $=$ MNW solution $=$
				 \begin{blockarray}{ccccccccccc}

				 		&&\matindex{$o_1$}&\matindex{$o_2$}&\\
				 	    \begin{block}{c(cccccccccc)}
				 			\matindex{$1$}& &$\nicefrac{1}{6}$&$1$ \\
				 			\matindex{$2$}& &$\nicefrac{5}{6}$&$0$ \\
				 	    \end{block}
				 	  \end{blockarray}
	
	{For CEEI, the price of $o_1$ is $6/5$ and $o_2$ is $4/5$.}

			Leximin solution $=$ Balanced Leximin $=$
		 \begin{blockarray}{ccccccccccc}

		 		&&\matindex{$o_1$}&\matindex{$o_2$}&\\
		 	    \begin{block}{c(cccccccccc)}
		 			\matindex{$1$}& &$0$&$1$ \\
		 			\matindex{$2$}& &$1$&$0$ \\
		 	    \end{block}
		 	  \end{blockarray}

			\end{example}

			Note that if all the outcomes are balanced (as is the case under the HZ rule) and an agent's preferences are dichotomous, then an agent's preferences over the allocations only depends on the stochastic dominance relation over outcomes.


		%
		%
		%
		%
		%
		%
		%
		%

			\section{Bi-valued Utilities: Relations of HZ with other Rules and Algorithms}

				Binary (1-0) utilities are a special class of utilities under which many rules and algorithms coincide.

			\begin{fact}\label{fact:1}
				Under 1-0 utilities, the following rules are  equivalent even if the number of items is different than the number of agents:

			\begin{enumerate}
				\item Leximin rule 
				\item Maximum Nash Welfare (MNW) rule
				\item Competitive Equilibrium with Equal Incomes (CEEI)~\citep{Vari74a} 
				\item Controlled Cake Eating Algorithm (CCEA)\footnote{The Controlled Cake Eating Algorithm (CCEA) algorithm~\citep{AzYe14a} and Mechanism 1 of \citet{CLPP10a} are described in the context of cake cutting.  They also apply to allocation of items: each cake segment can be treated as a separate item.}~\citep{AzYe14a} 
			\item Mechanism 1 of \citet{CLPP10a}.
			\end{enumerate}
			\end{fact}

			For example, CEEI and MNW are well-known to be equivalent even for general additive utilities~\citep{Vazi07a}. Under 1-0 utilities, all the rules were shown to be equivalent~\citep{AzYe14a}. 
			Since the leximin rule gives rise to a unique utility profile (agents' utilities do not change under different leximin outcomes), it follows that all the rules above give rise to a unique utility profile. 
			The rules above may not return a balanced assignment even for 1-0 utilities. Therefore, they are most suitable when the items are viewed as divisible.

			Next we highlight the intimate connection between the HZ rule under bi-valued utilities and the elegant Extended Probabilistic Serial (EPS) algorithm of \citet{KaSe06a}. 
			EPS is well-defined for any weak orders but we will stick to its presentation for the case of dichotomous preferences. The running time is $O(n^3\log n)$ under dichotomous preferences.  \citet{KaSe06a} note that EPS for dichotomous preferences is equivalent to the egalitarian rule studied by \citet{BoMo04a}.
			The egalitarian rule studied by \citet{BoMo04a} is the leximin rule applied to the set of balanced (unit-demand) assignemnts. \citet{BoMo04a} studied the rule in the context of two-sided matching with men on one side and women on the other side. 
			We adapt the presentation of \citet{KaSe06a} which is more algorithmic in nature and is directly focussed on the assignment problem with one-sided preferences. 

					\paragraph{Extended Probabilistic Serial (EPS) for dichotomous preferences} 
					For each agent $i$, let $D_i$ denote the set of items most preferred by $i$. Agents gradually guarantee more and more fractional amount of liked items until agents cannot guarantee more. At this point there is a bottleneck set of agents who cannot fractionally get more amount of liked items. 
					Let 
					$v=\min_{C\subseteq N}\frac{|\cup_{i\in C}D_i|}{|C|}$
					and $X_1$ denote the largest cardinality set $X_1\subseteq N$ for which $\frac{|\cup_{i\in X_1}D_i|}{|X_1|}=v$. Such a bottleneck set can be computed via network flows as explained by \citet{KaSe06a}. When we are allowed fractional allocations, then agents in $X_1$ can each get utility $v$. 
					These agents $X_1$ and the items that they like $O_1$ are removed from the market. The same process is recursively applied to the remaining market until all items are allocated. Once all agents exit from the market, then the remaining items are allocated among those agents who got less than one unit of items to ensure that the allocation is balanced. Note that each successive bottleneck set has a strictly higher utility guarantee $v$: $v_1<v_2<\cdots < v_k$. The algorithm is specified as Algorithm~\ref{algo:EPS}.

						\begin{algorithm}[h!]
								\caption{EPS rule for dichotomous preferences}
								\label{algo:EPS}
							\normalsize
							\begin{algorithmic}

						\REQUIRE $(N,O,u)$ under 1-0 utilities
						\ENSURE	A balanced assignment 									
					\end{algorithmic}
						\begin{algorithmic}[1]
							\normalsize
			
						\STATE $N'\longleftarrow N$; $O'\longleftarrow O$
						\STATE $k\longleftarrow 1$
						\WHILE{$N'\neq \emptyset$}
						\IF{each agent $i\in N'$ can get one unit of items in $D_i\cap O'$}
						\STATE Give each agent $i\in N'$ one unit of items from $D_i\cap O'$ [can be done via an algorithm to compute a maximum size matching]
						\ELSE
						\STATE 	 
								Let 
								\[v_k=\min_{C\subseteq N'}\frac{|\cup_{i\in C}(D_i\cap O')|}{|C|}\]
								and $X_k$ denote the largest cardinality set $X_k\subseteq N'$ for which $\frac{|\cup_{i\in X_k}(D_i\cap O')|}{|X_k|}=v_k$. Such a bottleneck set can be computed via network flows~\citep{KaSe06a}. Let $\cup_{i\in X_k}(D_i\cap O')$ be $O_k$.
								\STATE Agents in $X_k$ can each get utility $v_k$ by getting items from $O_k$. Assignment $x$ for agents in $X_k$ is finalized. These agents $X_k$ and the items that they like $O_k$ are removed from the market:
								$N'\longleftarrow N'\setminus X_k'$; $O'\longleftarrow O'\setminus O_k$; 
								\ENDIF
								\STATE $k\longleftarrow k+1$
						\ENDWHILE
								\STATE \label{step:giveback} For agents in $\{i\in N\midd \sum_{o_j\in O}x_{ij}<1\}$, give them any remaining unallocated items in $O^0\subset O'$ to ensure that the assignment $x$ is balanced.
								\RETURN $x$.
										
								\end{algorithmic}
							\end{algorithm}

		
			Next, we present a theorem which clarifies the relations between many rules.

			\begin{theorem}\label{th:charac}
				For 1-0 utilities, the following rules are equivalent. 
				\begin{enumerate}
							\item \label{1} HZ rule
					\item \label{2} EPS rule
							\item \label{3} Egalitarian rule of \citet{BoMo04a}
					\item \label{4} Leximin applied to the set of balanced assignments
					\item \label{5} MNW applied to the set of balanced assignments
					\item \label{6} The rule that applies the balancing operation to a leximin solution
					\item \label{7} The rule that applies the balancing operation to a MNW solution
					\item \label{8} The rule that applies the balancing operation to a CEEI solution.
				\end{enumerate}
				\end{theorem}
				\begin{proof}
					\noindent
					\begin{itemize}
						\item[] 
						\ref{2} $\implies$ \ref{1}:

						Let the EPS outcome be $x$.
						We compute the prices $p$ of the items such that 
						$x_i\in \{x_i'\in \mathcal{A}_b \midd x_i'\in \text{argmax}\{u_i(x_i'):\sum_{o_j\in O}x_{ij}'\cdot(p_j)\leq 1\}$.
						Consider the run of EPS on dichotomous preferences. When a set of agents $X_k\subseteq N$ in EPS becomes a bottleneck set and each agent gets utility $v$, we can set the the individual prices of the goods allocated to agents in $X_k$ to $p_j=1/v_k$ for all $o_j$ allocated to the agents in $X_k$. Each agent in $X_k$ at this point gets an allocation that does not exceed size constraints. Moreover, it gets total utility $v_k$ for items each of which cost $1/v_k$. An agent $i\in X_k$ does not like any items after items $O_k$ are removed from $O'$. It may most prefer items that were removed before $O_k$ were removed. However, those items have even higher prices because the $v$ value progressively becomes more with the next bottleneck set and hence the prices keep going lower. Therefore for each agent in $X_k$, the utility $v_k$ is the maximum utility that can be achieved if $i\in N$ was to buy liked items at their prices according to $p$. For any extraneous items that are allocated in Step~\ref{step:giveback}, they can get price $0$.
						We have proved that \ref{2} implies \ref{1}.

						\ref{1} $\implies$ \ref{8} 
				
						Consider an HZ solution. Note that no agent pays anything for a zero utility item because if it did, it can get a bit more of a one utility item. 
						Therefore, if an agent gets a zero utility item, its price is zero. 
			We show that there is a CEEI solution under the same prices. 				
			Consider all the agents who get utility less than 1 in an HZ solution. For such agents, the size constraint does not impact in specifying their demand set (best possible feasible allocations within the budget). Now consider the agents who get utility 1 in the HZ solution. Their demand set changes when the size constraints are removed because they can avail an additional amount of items. In the HZ solution, 
			all such additional items are given to agents with utility less than one to ensure than every agent has total amount one. Hence, these additional items have zero price. We claim that there is a CEEI outcome under the same prices. Each agent who gets utility less than one in the HZ outcome clearly maximizes her utility even if the size constraints are removed. The only agents who can get more utility under the CEEI outcomes are the ones who got utility one in the HZ solution but can benefit from zero price items. Therefore the market clears under the same prices even if there are no size constraints. From Fact~\ref{fact:1}, it follows all agents that get utility less than 1 in an HZ solution, will get the same utility in a CEEI/leximin/MNW solution. As for agents who get utility 1 in an HZ solution, they cannot get any more utility in a balanced assignment. 
				

						\item[] \ref{2} $\iff$  \ref{3} $\iff$ \ref{4} $\iff$ \ref{5}:
			
								The EPS rule of \citet{KaSe06a} returns a balanced random assignment. Under dichotomous preferences, the EPS rule has a direct connection with the egalitarian rule proposed by \citet{BoMo04a} for two-sided matching problem with equal number of men and women with dichotomous preferences. If men are treated as items who are completely indifferent among women, then the setting studied by \citet{BoMo04a} reduces to the random assignment problem with one side having dichotomous preferences and the EPS rule coincides with the egalitarian rule of \citet{BoMo04a}. The egalitarian rule of \citet{BoMo04a} is equivalent to the leximin rule on the set of balanced assignments. \citet{BoMo04a} also note that their egalitarian rule is equivalent to the MNW rule applied to the set of balanced assignments. Thus the equivalence between \ref{2}, \ref{3}, \ref{4}, and \ref{5} follows from the papers of \citet{BoMo04a} and  \citet{KaSe06a}.
						\item[] \ref{6} $\iff$ \ref{7} $\iff$ \ref{8}:
			
							We also know from Fact~\ref{fact:1} that under 1-0 utilities, MNW, CEEI, and leximin coincide. Therefore, \ref{6}, \ref{7}, and \ref{8} are equivalent. 
			
							\item[] \ref{2} $\iff$ \ref{6}:
				
							Consider an assignment $x$ that is balanced and leximin among balanced assignments. Suppose it is not globally leximin. Then observe how a leximin assignment among balanced assignments is achieved during the run of the EPS algorithm. For all agents in bottleneck sets who get utility less than 1, their utilities are exactly the same as they would get in a globally leximin random assignment. The reason $x$ is not a globally leximin assignment is that the last set of agents who exit the market get utility 1 but some of them could have got utility strictly more than 1 (without decreasing the utility of other agents) if the balancedness condition is not imposed. These items are not additionally allocated to the agents who already have utility 1 and these items are only distributed in Step~\ref{step:giveback} of Algorithm~\ref{algo:EPS}. It follows that assignment $x$ can be achieved by first computing a balanced assignment that is leximin and then implementing the balancing operation on it. 
					\end{itemize}

				This completes the proof. 
								\end{proof}
	
			%
			%
			%
			%

				A corollary of the theorem above is the following one. 
	
				\begin{corollary}
					For all HZ solutions under 1-0 utilities, each agent gets the same utility in all the solutions. 
					\end{corollary}
					\begin{proof}
						We proved that under 1-0 utilities, the HZ solution is equivalent to applying the balancing operation to a CEEI solution. The utilities of each agent are invariant under all CEEI solutions. Hence, it follows that the utilities of each agent are invariant under all HZ solutions~\citep{Vazi07a}. 
						\end{proof}
 
			 Next, we present the following theorem for the case of bi-valued utilities.
	
				\begin{theorem}\label{th:th2}
					Under bi-valued utilities, the HZ rule is equivalent to
			\begin{enumerate}
				\item EPS rule
			\item  Egalitarian rule of \citet{BoMo04a} applied with respect to the binary-reduced utilities
					\item  Leximin applied with respect to the binary-reduced utilities to the set of balanced assignments
					\item  MNW applied with respect to the binary-reduced utilities to the set of balanced assignments
					\item The rule that applies the balancing operation to a leximin solution (with respect to the binary-reduced utilities)
					\item  The rule that applies the balancing operation to a MNW solution (with respect to the binary-reduced utilities) 
					\item  The rule that applies the balancing operation to a CEEI solution (with respect to the binary-reduced utilities).
					\end{enumerate}
					\end{theorem}
							\begin{proof}
								We know from Theorem~\ref{th:charac} that HZ and EPS are equivalent under 1-0 utilities. Note that EPS is an ordinal algorithm so it gives the same outcome under dichotomous preferences. The same argument that is used to prove that EPS gives an HZ solution under 1-0 utilities can be used verbatim to prove that EPS gives an HZ solution under bi-valued utilities. At any point at which a bottleneck set $X_k$ is removed, the agents in the set are able to get their most preferred items at the lowest price. Any most preferred items that are not available were sold at a higher price. As for lesser preferred items, they are given to the agents for free because they are the under-demanded items that are given price zero. 
					
								We have shown that HZ and EPS are equivalent under 1-0 utilities;
								EPS solutions are equivalent under 1-0 utilities and bi-valued utilities; and EPS under bi-valued utilities gives a HZ solution under bi-valued utilities. It follows that HZ under 1-0 utilities implies HZ under bi-valued utilities.\footnote{The proof is also an alternative argument that 	HZ under bi-valued utilities is invariant under scaling and shifting of the utility functions as their outcomes are equivalent to HZ under 1-0 utilities.} 
					
								Next, we prove that any HZ solution under bi-valued utilities is the HZ solution under binary-reduced utilities for the same item prices. Suppose there is an HZ solution under bi-valued utilities that is not an outcome of HZ under binary-reduced utilities. This means that the market does not clear under binary-reduced utilities for the same prices. For the base case, consider the items in $O_1$ in the corresponding EPS outcome. The agents $X_1$ pay nothing for the lesser preferred items. If some item in $O_1$ has a different price, then for the market to clear, at least some item in $O_1$ has lesser price. 
								Consider the item $o\in O_1$ whose price dropped the most. But then all agents in $X_1$ who most prefer $o$ want to get more of $o$ which implies that the market does not clear. The same argument works inductively for the items in $O_2,\ldots, O_k$. Hence, the market clears for binary-reduced utilities for the original prices which contradicts that the HZ solution under bi-valued utilities that is not an outcome of HZ under binary-reduced utilities.

							%
							%
							%
							%
							%

													%
								%
					
		
					The remaining equivalences follow from Theorem~\ref{th:charac} that EPS is equivalent to the rules under 1-0 utilities.
								\end{proof}
		
			
			%

				\begin{corollary}
					For all HZ solutions under bi-valued utilities, the solution is rational. For a given problem instance with bi-valued utilities, each agent gets the same utility in all the  solutions of the problem instance. 
			Under bi-valued utilities, the set of HZ solutions does not change even if the agents' utilities are shifted or scaled. 
					\end{corollary}
					\begin{proof}
						We have proved that HZ solutions under bi-valued utilities is equivalent to applying the balancing operation to a CEEI solution (with respect to the binary-reduced utilities). It is well-known that the CEEI solution is always rational and, for the problem instance, each agent gets the same utility in all the solutions of the instance, and each item gets the same price in all the solutions of the instance. 
						\end{proof}

			
			The theorem above also has several algorithmic consequences. 
			Firstly, in order to compute the HZ solution for bi-valued utilities, one only needs to consider the underlying dichotomous preferences and run the EPS algorithm. The characterizations combined with the EPS algorithm provide an alternative route to proving that the HZ solution can be computed in strongly polynomial time. 

			\begin{corollary}
				Under bi-valued utilities, the HZ solution can be computed in time $O(n^3\log n)$.
				\end{corollary}
				\begin{proof}
					Computing the HZ solution reduces to computing the outcome of the EPS algorithm. EPS runs in $O(n^3\log n)$ time.
					\end{proof}

			Another algorithmic consequence is a reduction from HZ to CEEI in the case of bi-valued utilities. 

			\begin{corollary}
				There is a linear-time reduction from computing the HZ solution under bi-valued utilities to computing the MNW / CEEI / leximin solution under 1-0 utilities. 
				\end{corollary}
				\begin{proof}
					The reduction is specified as Algorithm~\ref{algo:reduction}.
					\end{proof}
	
	Our approach for turning HZ into a leximin problem that is oblivious to tracking prices can be seen as conceptually simpler.


						\begin{algorithm}[h!]
								\caption{Reduction from HZ for bi-valued utilities to MNW/CEEI/leximin}
								\label{algo:reduction}
									\normalsize
							\begin{algorithmic}

						\REQUIRE $(N,O,u)$ where $u$ is bi-valued
						\ENSURE	A balanced assignment 								
					\end{algorithmic}
						\begin{algorithmic}[1]
							\normalsize
						\STATE Turn the bi-valued utilities $u$ to binary reduced utilities $u'$.
						\STATE Apply an algorithm for the (unconstrained) MNW/CEEI/lexi-min to utilities $u'$ to compute a solution $x$. 
						\STATE Apply the balancing operation on $x$ to derive the HZ outcome.
 									
							\RETURN $x$.
										
							\end{algorithmic}
						\end{algorithm}
															
						We have already pointed out that in EPS, an agent who gets one unit of most-preferred items cannot get more even though he may be the only agent liking the additional item. Due to the adherence to the size constraints, the HZ solution may not be ex-ante Pareto optimal (PO) among the set of \emph{all} assignments if we assume that agents have additive utilities. 
		Assuming, we ask the agents to get the cheapest demand set, the HZ solution is PO among the set of balanced assignments. It may not be PO among the set of all assignments if we assume that agents have additive utilities.				
					


								We recall that CEEI is equivalent to the MNW rule. However, the HZ rule is not equivalent to maximizing Nash welfare while imposing equal size constraints. This is evident from  Example~\ref{example:main} that showed that applying MNW to balanced assignments may give highly unfair assignments.

			\section{Strategyproofness}

			For all the rules that we have considered, we discuss 
			issues around strategyproofness and group strategyproofness. 
			A rule is \emph{strategyproof} if no agent can misreport her preferences to get higher utility.
			A rule is ex-ante \emph{group}    
			\emph{strategyproof} if no group of agents can misreport their preferences so that all agents in the group get at least as much utility and at least one agent in the group gets strictly more utility.
 
			For dichotomous preferences, EPS is ex-ante group-strategyproof. This fact has been shown before as well (see, e.g. Theorem 1 of \citet{KaSe06a} who refer to the argument by \citet{BoMo04a}).
			Group-strategyproofness for EPS is established by induction on the bottleneck sets created: it can be proved that no agent in a bottleneck set will be a member of a manipulating coalition. 
			Hence, it follows that the HZ solution is also group-strategyproof under bi-valued utilities.\footnote{The HZ rule is not strategyproof for general utilities. See, for example, further discussion by \citet{ACGH20a}.}

			\begin{corollary}
				Under bi-valued utilities, the HZ rule is group-strategyproof.
				\end{corollary}

			For 1-0 utilities, all the rules leximin/MNW/CEEI that may return unbalanced assignments are group-strategyproof as well. See for example Theorem 3 of \citet{AzYe14a} that shows that leximin/MNW/CEEI are group-strategyproof. Within the class of bi-valued utilities, it is clear that leximin is not strategyproof or envy-free: 
			an agent with scaled-down utilities gets predominant importance under the leximin rule. On the other hand, scaling down of utilities has no effect in the case of MNW (equivalently CEEI). Despite resistance to manipulation by scaling, we show that MNW is not strategyproof under bi-valued utilities even if the underlying ordinal preferences remain unchanged. Previously, manipulations of the MNW and CEEI under general additive utilities have been considered in many papers~\citep{BCD+1a,CDT+16a}.

			\begin{theorem}\label{CEEI_SP}
				Under bi-valued utilities, MNW/CEEI is not strategyproof even if the underlying ordinal preferences remain unchanged. 
				\end{theorem} 

				\begin{proof}
					We provide an example with 5 agents and 5 items. 
					\begin{center}
							\setlength{\tabcolsep}{6pt}
							\begin{tabular}{ccccccccc}
								& $o_1$ & $o_2$& $o_3$ & $o_4$ & $o_5$\\
								\midrule
								$1$ &10&1&1&1&1\\
								$2$& 6&6&10&6&6\\
								$3$& 4&10&4&10&4\\
								$4$& 0&0&0&0&1\\ 
								$5$& 0&0&0&0&1\\
							\end{tabular}
						\end{center}
			
						Under the valuations, the MNW/CEEI outcome $x$ is as follows. 
			

					\begin{center}
									 \begin{blockarray}{ccccccccccc}
				&&\matindex{$o_1$}&\matindex{$o_2$}& \matindex{$o_3$}& \matindex{$o_4$}& \matindex{$o_5$}& \\
									 	    \begin{block}{c(cccccccccc)}
				\matindex{$1$}&&$1$& $0$&$0$&$0$&$0$&\\ 
				\matindex{$2$}&&$0$&$\nicefrac{1}{12}$&$1$&$\nicefrac{1}{12}$&$0$&\\
				\matindex{$3$}&&$0$&$\nicefrac{11}{12}$&$0$&$\nicefrac{11}{12}$&$0$&\\
				\matindex{$4$}&&$0$& $0$&$0$&$0$&$\nicefrac{1}{2}$&\\ 
				\matindex{$5$}&&$0$& $0$&$0$&$0$&$\nicefrac{1}{2}$&\\
				\end{block}
									 	  \end{blockarray}
										  \end{center}
				%
			
						Agent 1 gets utility 10. 
			
						Suppose agent 1 misreports as follows by raising her value for the lower preferred items. 
				\begin{center}
						\setlength{\tabcolsep}{6pt}
						\begin{tabular}{ccccccccc}
							& $o_1$ & $o_2$& $o_3$ & $o_4$ & $o_5$\\
							\midrule
							$1$ &10&8&8&8&8\\
							$2$& 6&6&10&6&6\\
							$3$& 4&10&4&10&4\\
							$4$& 0&0&0&0&1\\ 
							$5$& 0&0&0&0&1\\
						\end{tabular}
					\end{center}
		
							Under the misreport, the MNW/CEEI outcome $y$ is as follows. 
				
								\begin{center}
												 \begin{blockarray}{ccccccccccc}
							&&\matindex{$o_1$}&\matindex{$o_2$}& \matindex{$o_3$}& \matindex{$o_4$}& \matindex{$o_5$}& \\
												 	    \begin{block}{c(cccccccccc)}
							\matindex{$1$}&&$1$& $\nicefrac{3}{16}$&$0$&$\nicefrac{3}{16}$&$0$&\\ 
							\matindex{$2$}&&$0$&$0$&$1$&$0$&$0$&\\
							\matindex{$3$}&&$0$&$\nicefrac{13}{16}$&$0$&$\nicefrac{13}{16}$&$0$&\\
							\matindex{$4$}&&$0$& $0$&$0$&$0$&$\nicefrac{1}{2}$&\\ 
							\matindex{$5$}&&$0$& $0$&$0$&$0$&$\nicefrac{1}{2}$&\\
							\end{block}
												 	  \end{blockarray}
													  \end{center}

				The assignment $y$ gives agent 1 utility more than 10 with respect to her original utilities. Hence, MNW/CEEI is not strategyproof under bi-valued utilities. The computation of MNW solutions is via an optimisation solver. 
				An analytical proof that the computation is correct is given below. 	Let
	\[
		f_u(x) = \prod_{i \in N}\sum_{o_j \in O} u_{ij} x_{ij},
	\]
	and define constraints (for all \(i \in N\) and \(o_j \in O\))
	\begin{align*}
		c_j(x) &= \sum_{i^\prime \in N} x_{i^\prime j} - 1, &
		c_{ij}(x) &= -x_{ij}.
	\end{align*}
	Write
	\[
		u_{-i}(x) = \prod_{k \in N \setminus \{i\}}\sum_{o_j \in O} u_{k j} x_{k j}
	\]
	for the product of utilities, ignoring agent \(i\). Then \(\partial_{(ij)} f_u = u_{ij} u_{-i}\).

	For the first example (truthful utilities), we find
	\begin{align*}
		u_{-1}(x) &= \nicefrac{605}{12}, &
		u_{-2}(x) &= \nicefrac{275}{6}, &
		u_{-3}(x) &= \nicefrac{55}{2}, &
		u_{-4}(x) = u_{-5}(x) &= \nicefrac{3025}{3},
	\end{align*}
	from which
	\[
		\nabla f_u(x) = \begin{pmatrix}
			\nicefrac{3025}{6} & \nicefrac{605}{12} & \nicefrac{605}{12} & \nicefrac{605}{12} & \nicefrac{605}{12} \\
			275 & 275 & \nicefrac{1375}{3} & 275 & 275 \\
			110 & 275 & 110 & 275 & 110 \\
			0 & 0 & 0 & 0 & \nicefrac{3025}{3} \\
			0 & 0 & 0 & 0 & \nicefrac{3025}{3}
		\end{pmatrix}.
	\]
	The gradient is a vector, but we write it in matrix form so we can directly relate each \(x_{ij}\) partial with the corresponding matrix entry. To find Lagrange multipliers and apply the KKT conditions, we only need to find the gradients of the active constraints. \(c_j\) for \(1 \leq j \leq 5\) are active, so we find \([\nabla c_j(x)]_{ki} = \mathbb{I}_{k = j}\). That is, \(\nabla c_j(x)\) is zero everywhere, except for every entry in column \(j\), which are all \(1\). Then, we can calculate part of the Lagrangian,
	\begin{alignat*}{2}
		M &=&&\nabla f_u(x) - \nicefrac{3025}{6} \nabla c_1(x) - 275 \nabla c_2(x) \\
		&&&- \nicefrac{1375}{3} \nabla c_3(x) - 275 \nabla c_4(x) - \nicefrac{3025}{3} \nabla c_5(x) \\
		&=&&\begin{pmatrix}
			0 & -\nicefrac{2695}{12} & -\nicefrac{4895}{12} & -\nicefrac{2695}{12} & -\nicefrac{11495}{12} \\
			-\nicefrac{1375}{6} & 0 & 0 & 0 & -\nicefrac{2200}{3} \\
			-\nicefrac{2365}{6} & 0 & -\nicefrac{1045}{3} & 0 & -\nicefrac{2695}{3} \\
			-\nicefrac{3025}{6} & -275 & -\nicefrac{1375}{3} & -275 & 0 \\
			-\nicefrac{3025}{6} & -275 & -\nicefrac{1375}{3} & -275 & 0
		\end{pmatrix}.
	\end{alignat*}
	Now the remaining active constraints are precisely those \(c_{ij}\) with corresponding negative entries in this matrix, and since \(\nabla c_{ij} = -\mathbf{e}_{ij}\), we can choose Lagrange multipliers for these extra gradients such that the Lagrangian gradient is zero, and these multipliers are negative. Specifically, the Lagrange multiplier for constraint \(c_{ij}\) is exactly \(M_{ij}\). Since \(f_u\) is concave over the constrained set, which is convex, we have a concave optimisation problem. All Lagrangian multipliers of active inequality constraints are negative, so we satisfy the KKT conditions, and can conclude that \(x\) is a global maximum of \(f_u\).

	For the second example (misreported utilities), we find
	\begin{align*}
		u_{-1}(y) &= \nicefrac{325}{8}, &
		u_{-2}(y) &= \nicefrac{845}{16}, &
		u_{-3}(y) &= \nicefrac{65}{2}, &
		u_{-4}(y) = u_{-5}(y) &= \nicefrac{4225}{4},
	\end{align*}
	so that
	\[
		\nabla f_u(y) = \begin{pmatrix}
			\nicefrac{1625}{4} & 325 & 325 & 325 & 325 \\
			\nicefrac{2535}{8} & \nicefrac{2535}{8} & \nicefrac{4225}{8} & \nicefrac{2535}{8} & \nicefrac{2535}{8} \\
			130 & 325 & 130 & 325 & 130 \\
			0 & 0 & 0 & 0 & \nicefrac{4225}{4} \\
			0 & 0 & 0 & 0 & \nicefrac{4225}{4}
		\end{pmatrix}.
	\]
	Then
	\begin{alignat*}{2}
		M &=&&\nabla f_u(y) - \nicefrac{1625}{4} \nabla c_1(y) - 325 \nabla c_2(y) \\
		&&&- \nicefrac{4225}{8} \nabla c_3(y) - 325 \nabla c_4(y) - \nicefrac{4225}{4} \nabla c_5(y) \\
		&=&&\begin{pmatrix}
			0 & 0 & -\nicefrac{1625}{8} & 0 & -\nicefrac{2925}{4} \\
			-\nicefrac{715}{8} & -\nicefrac{65}{8} & 0 & -\nicefrac{65}{8} & -\nicefrac{5915}{8} \\
			-\nicefrac{1105}{4} & 0 & -\nicefrac{3185}{8} & 0 & -\nicefrac{3705}{4} \\
			-\nicefrac{1625}{4} & -325 & -\nicefrac{4225}{8} & -325 & 0 \\
			-\nicefrac{1625}{4} & -325 & -\nicefrac{4225}{8} & -325 & 0
		\end{pmatrix}.
	\end{alignat*}
	Similar analysis gives that \(y\) is a global maximiser of \(f_u\).
				\end{proof}

					The theorem above can be recast in the context of indivisible goods to state that for bi-valued utilities, the MNW rule of \citet{CKM+16a} is not strategyproof for any tie-breaking over the set of possible outcomes. The MNW rule of \citet{CKM+16a} for indivisible goods coincides with MNW for divisible goods if the goods are made arbitrarily small. 
		
					\begin{corollary}
				Under bi-valued utilities, and for indivisible goods, the MNW rule of \citet{CKM+16a} is not strategyproof for any tie-breaking over the set of possible outcomes. 
						\end{corollary}
		
					The argument follows from the observation that
			each divisible good can be approximately modelled as small enough multiple indivisible goods.

			\section{The Curious Case of the Nash Bargaining Rule}
		
		In this section, we consider the Nash Bargaining rule that is very similar to several rules analyzed in the paper including MNW, CEEI, and HZ. We show that despite its close connections with these rules, it does not exhibit the same properties even under binary utilities. 
			A \emph{Nash bargaining (NB)} solution with disagreement point \(\{d_i\}_{i \in N}\) is a solution $x^*\in \mathcal{F}_b$ to the optimisation problem
			\begin{alignat*}{2}
			    \max_{x\in \mathcal{F}_b} & \mathrlap{\ \prod_{i \in N} \left(\sum_{o_j \in O} u_{ij} x_{ij} - d_i\right);} && \\
			    \text{s.t.} &\quad& d_i - \sum_{o_j \in O} u_{ij} x_{ij} &\leq 0, \\
				&& \sum_{i \in N} x_{ij} - 1 &\leq 0, \\
			    && -x_{ij} &\leq 0.
			\end{alignat*}
			We will use the disagreement point
			\[
				d_i = \frac{1}{n} \sum_{o_j \in O} u_{ij},
			\]
			referred to as uniform disagreement, as it comes from utilities achieved from a uniform allocation. The uniform agreement is standard (see e.g., \citep{ACGH20a}). \citet{ACGH20a} refer to NB as the benchmark for strong efficiency and fairness guarantees. One appealing aspect of NB is that it is invariant to shifting and scaling of reported utilities. 
		
			We first observe the NB is very similar in spirit to MNW, CEEI, and HZ. It is the same as MNW except that it takes into account the uniform disagreement point and restricts its attention to balanced assignments. We have already discussed the close links between MNW with CEEI and HZ. 
	Next we show that whereas MNW (equivalently CEEI) and HZ are envy-free and strategyproof for 1-0 utilities, NB fails both properties. In order to prove the results, we use the following characterization of NB that holds under 1-0 utilities. The characterization states that, under 1-0 utilities, computing NB is equivalent to leximin optimisation with respect to the differences between agent utilities and their disagreement points. 

	\begin{lemma}
		Under 1-0 utilities, any balanced allocation that is leximin-optimal with respect to the disagreement point is an NB solution.
	\end{lemma}
	\begin{proof}
		It is known that under 1-0 utilities, an improvement in the Nash welfare results in a leximin improvement and a leximin improvement results in an improvement in the Nash welfare (see e.g., \citep{HPPS20a}). The fact holds if we interpret each agent's utility as the difference between her real utility and her disagreement point utility. 
	\end{proof}


	Next, we prove that NB is not strategyproof even for 1-0 utilities. The fact that it is not strategyproof for general cardinal utilities was already known (see for example the discussion of \citet{ACGH20a}). Our key insight for the case of 1-0 utilities is that agents  can benefit from expanding their set of liked items to artificially get a higher demand. Then their disagreement point increases, ensuring they get more allocated to them, and the item they do not actually like is higher-demanded, so the excess leans towards the item they like instead.
		
			%
			%

	\begin{theorem}\label{th:sp-nb}
		Even under 1-0 utilities, the Nash bargaining solution is not strategyproof.
	\end{theorem}
	\begin{proof}
		Suppose truthful utilities were given by
		\begin{center}
			\setlength{\tabcolsep}{6pt}
			\begin{tabular}{ccccccccc}
				& \(o_1\) & \(o_2\) & \(o_3\) & \(o_4\) & \(o_5\) \\
				\midrule
				\(1\) & 1 & 0 & 0 & 0 & 0 \\
				\(2\) & 1 & 0 & 0 & 0 & 0 \\
				\(3\) & 0 & 1 & 0 & 0 & 0 \\
				\(4\) & 0 & 1 & 0 & 0 & 0 \\ 
				\(5\) & 0 & 0 & 1 & 1 & 1 \\
			\end{tabular}
		\end{center}
		but agent 2 misreports so the reported utility profile is actually
		\begin{center}
			\setlength{\tabcolsep}{6pt}
			\begin{tabular}{ccccccccc}
				& \(o_1\) & \(o_2\) & \(o_3\) & \(o_4\) & \(o_5\) \\
				\midrule
				\(1\) & 1 & 0 & 0 & 0 & 0 \\
				\(2\) & 1 & 1 & 0 & 0 & 0 \\
				\(3\) & 0 & 1 & 0 & 0 & 0 \\
				\(4\) & 0 & 1 & 0 & 0 & 0 \\ 
				\(5\) & 0 & 0 & 1 & 1 & 1 \\
			\end{tabular}
		\end{center}
		Under these profiles, NB solutions respectively are
		\begin{center}
			\(x = \) \begin{blockarray}{ccccccccccc}
				&&\matindex{$o_1$}&\matindex{$o_2$}& \matindex{$o_3$}& \matindex{$o_4$}& \matindex{$o_5$}& \\
					\begin{block}{c(cccccccccc)}
					\matindex{$1$}& & $\nicefrac{1}{2}$ & $0$ & $0$ & $\nicefrac{1}{2}$ & $0$ & \\ 
					\matindex{$2$}& & $\nicefrac{1}{2}$ & $0$ & $0$ & $\nicefrac{1}{2}$ & $0$ & \\
					\matindex{$3$}& & $0$ & $\nicefrac{1}{2}$ & $0$ & $0$ & $\nicefrac{1}{2}$ & \\
					\matindex{$4$}& & $0$ & $\nicefrac{1}{2}$ & $0$ & $0$ & $\nicefrac{1}{2}$ & \\ 
					\matindex{$5$}& & $0$ & $0$ & $1$ & $0$ & $0$ & \\
				\end{block}
			\end{blockarray},
		\end{center}
		and 
		\begin{center}
			\(x^* = \) \begin{blockarray}{ccccccccccc}
				&&\matindex{$o_1$}&\matindex{$o_2$}& \matindex{$o_3$}& \matindex{$o_4$}& \matindex{$o_5$}& \\
					\begin{block}{c(cccccccccc)}
					\matindex{$1$}& & $\nicefrac{9}{20}$ & $0$ & $0$ & $\nicefrac{11}{20}$ & $0$ & \\ 
					\matindex{$2$}& & $\nicefrac{11}{20}$ & $\nicefrac{1}{10}$ & $0$ & $\nicefrac{7}{20}$ & $0$ & \\
					\matindex{$3$}& & $0$ & $\nicefrac{9}{20}$ & $0$ & $\nicefrac{1}{10}$ & $\nicefrac{9}{20}$ & \\
					\matindex{$4$}& & $0$ & $\nicefrac{9}{20}$ & $0$ & $0$ & $\nicefrac{11}{20}$ & \\ 
					\matindex{$5$}& & $0$ & $0$ & $1$ & $0$ & $0$ & \\
				\end{block}
			\end{blockarray}.
		\end{center}
		Under truthful reporting, agent 2 gets utility of \(\nicefrac{1}{2}\), but under the misreported utilities, they get utility of \(\nicefrac{11}{20}\). Thus NB is not strategyproof, even under 1-0 utilities.
	
		It remains to show that \(x\) and \(x^*\) are actually NB solutions. For \(x\), this is simple. The lexicographic utilities with respect to the disagreement point are
		\[
			\left(\nicefrac{3}{10}, \nicefrac{3}{10}, \nicefrac{3}{10}, \nicefrac{3}{10}, \nicefrac{7}{10}\right).
		\]
		Since agents 1 and 2 only like item 1, and have exhausted its demand, any increase in either's utility will decrease that of the other agent. Similarly for agents 3 and 4. Agent 5 has reached the demand constraint, so they cannot increase their utility at all. Thus, the allocation is leximin-optimal with respect to the disagreement point, and hence an NB solution.

		For \(x^*\), our lexicographic utilities with respect to the disagreement point are
		\[
			\left(\nicefrac{1}{4}, \nicefrac{1}{4}, \nicefrac{1}{4}, \nicefrac{1}{4}, \nicefrac{7}{10}\right).
		\]
		Since object 1 and 2 are both exhausted of supply:
		\begin{enumerate}
			\item Attempting to increase agent 1's utility will decrease agent 2's.
			\item Attempting to increase agent 3 or 4's utility will decrease the other's or agent 2's.
			\item Attempting to increase agent 2's utility will decrease agent 1, 3, or 4's utility.
		\end{enumerate}
		Again, agent 5 has reached the demand constraint, so they cannot increase their utility at all. Thus, the allocation is leximin-optimal with respect to the disagreement point, and hence an NB solution.
	\end{proof}

	Next we prove that NB violates envy-freeness even under 1-0 utilities. The result is especially surprising because MNW and HZ both have close links with NB, but do satisfy envy-freeness, even for general cardinal utilities. 
		
			\begin{theorem}
				Even under 1-0 utilities, a Nash bargaining solution is not always envy-free.
			\end{theorem}
	\begin{proof}
		In the proof of Theorem~\ref{th:sp-nb}, agent 1 envies agent 2 in the allocation \(x^*\).
	\end{proof}
		
		We expect that the negative results that we have proved for NB also carry over if we consider any other disagreement point.

			\section{Conclusions}
		
			The assignment problem is one of the most fundamental and widely-encountered allocation problems. In this paper, we provided a deeper understanding of the strong links as well as subtle differences between various well-known rules encountered in the literature. 
	One of our main contributions is a unification of the literature by proving the equivalence of several rules under binary and bi-valued utilities. We also show that the well-known Nash bargaining rule fails envy-freeness even under 1-0 utilities which was not known in the literature (to the best of our knowledge).	We envisage further work on the topic both in terms of algorithm design as well as further understanding of the tradeoffs between axiomatic properties.

	\bibliographystyle{plainnat}


\begin{thebibliography}{25}
	\providecommand{\natexlab}[1]{#1}
	\providecommand{\url}[1]{\texttt{#1}}
	\expandafter\ifx\csname urlstyle\endcsname\relax
	  \providecommand{\doi}[1]{doi: #1}\else
	  \providecommand{\doi}{doi: \begingroup \urlstyle{rm}\Url}\fi

	\bibitem[Abebe et~al.(2020)Abebe, Cole, Gkatzelis, and Hartline]{ACGH20a}
	R~Abebe, R.~Cole, V.~Gkatzelis, and J.~D. Hartline.
	\newblock A truthful cardinal mechanism for one-sided matching.
	\newblock In \emph{Proceedings of the 2020 {ACM-SIAM} Symposium on Discrete
	  Algorithms}, pages 2096--2113, 2020.

	\bibitem[Aziz and Ye(2014)]{AzYe14a}
	H.~Aziz and C.~Ye.
	\newblock Cake cutting algorithms for piecewise constant and piecewise uniform
	  valuations.
	\newblock In \emph{Proceedings of the 10th International Workshop on Internet
	  and Network Economics (WINE)}, pages 1--14, 2014.

	\bibitem[Aziz et~al.(2013)Aziz, Brandt, and Stursberg]{ABS13a}
	H.~Aziz, F.~Brandt, and P.~Stursberg.
	\newblock On popular random assignments.
	\newblock In \emph{Proceedings of the 6th International Symposium on
	  Algorithmic Game Theory (SAGT)}, volume 8146 of \emph{Lecture Notes in
	  Computer Science (LNCS)}, pages 183--194. Springer-Verlag, 2013.

	\bibitem[Aziz et~al.(2014)Aziz, Brandt, Brill, and Mestre]{ABBM14a}
	H.~Aziz, F.~Brandt, M.~Brill, and J.~Mestre.
	\newblock Computational aspects of random serial dictatorship.
	\newblock \emph{ACM SIGecom Exchanges}, 13\penalty0 (2):\penalty0 26--30, 2014.

	\bibitem[Aziz et~al.(2017)Aziz, Brill, Conitzer, Elkind, Freeman, and
	  Walsh]{ABC+16a}
	H.~Aziz, M.~Brill, V.~Conitzer, E.~Elkind, R.~Freeman, and T.~Walsh.
	\newblock Justified representation in approval-based committee voting.
	\newblock \emph{Social Choice and Welfare}, 48\penalty0 (2):\penalty0 461--485,
	  2017.

	\bibitem[Barman et~al.(2018)Barman, Krishnamurthy, and Vaish]{BKV18a}
	S.~Barman, S.~K. Krishnamurthy, and R.~Vaish.
	\newblock Greedy algorithms for maximizing {N}ash social welfare.
	\newblock In \emph{Proceedings of the 17th International Conference on
	  Autonomous Agents and Multi-Agent Systems (AAMAS)}, 2018.

	\bibitem[Bogomolnaia and Moulin(2001)]{BoMo01a}
	A.~Bogomolnaia and H.~Moulin.
	\newblock A new solution to the random assignment problem.
	\newblock \emph{Journal of Economic Theory}, 100\penalty0 (2):\penalty0
	  295--328, 2001.

	\bibitem[Bogomolnaia and Moulin(2004)]{BoMo04a}
	A.~Bogomolnaia and H.~Moulin.
	\newblock Random matching under dichotomous preferences.
	\newblock \emph{Econometrica}, 72\penalty0 (1):\penalty0 257--279, 2004.

	\bibitem[Bogomolnaia et~al.(2005)Bogomolnaia, Moulin, and Stong]{BMS05a}
	A.~Bogomolnaia, H.~Moulin, and R.~Stong.
	\newblock Collective choice under dichotomous preferences.
	\newblock \emph{Journal of Economic Theory}, 122\penalty0 (2):\penalty0
	  165--184, 2005.

	\bibitem[Brams and Fishburn(2007)]{BrFi07c}
	S.~J. Brams and P.~C. Fishburn.
	\newblock \emph{Approval Voting}.
	\newblock Springer-Verlag, 2nd edition, 2007.

	\bibitem[Br{\^{a}}nzei et~al.(2014)Br{\^{a}}nzei, Chen, Deng, Filos{-}Ratsikas,
	  Frederiksen, and Zhang]{BCD+1a}
	S.~Br{\^{a}}nzei, Y.~Chen, X.~Deng, A.~Filos{-}Ratsikas, S{\o}ren
	  Kristoffer~Stiil Frederiksen, and Jie Zhang.
	\newblock The fisher market game: Equilibrium and welfare.
	\newblock In \emph{Proceedings of the Twenty-Eighth {AAAI} Conference on
	  Artificial Intelligence, July 27 -31, 2014, Qu{\'{e}}bec City, Qu{\'{e}}bec,
	  Canada}, pages 587--593, 2014.

	\bibitem[Budish(2011)]{Budi11a}
	E.~Budish.
	\newblock The combinatorial assignment problem: Approximate competitive
	  equilibrium from equal incomes.
	\newblock \emph{Journal of Political Economy}, 119\penalty0 (6):\penalty0
	  1061--1103, 2011.

	\bibitem[Caragiannis et~al.(2016)Caragiannis, Kurokawa, Moulin, Procaccia,
	  Shah, and Wang]{CKM+16a}
	I.~Caragiannis, D.~Kurokawa, H.~Moulin, A.~D. Procaccia, N.~Shah, and J.~Wang.
	\newblock {The Unreasonable Fairness of Maximum Nash Welfare}.
	\newblock In \emph{Proceedings of the 17th ACM Conference on Economics and
	  Computation (ACM-EC)}, pages 305--322, 2016.

	\bibitem[Chen et~al.(2016)Chen, Deng, Tang, and Zhang]{CDT+16a}
	N.~Chen, X.~Deng, B.~Tang, and H.~Zhang.
	\newblock Incentives for strategic behavior in fisher market games.
	\newblock In \emph{Proceedings of the Thirtieth {AAAI} Conference on Artificial
	  Intelligence, February 12-17, 2016, Phoenix, Arizona, {USA}}, pages 453--459,
	  2016.
	\newblock URL
	  \url{http://www.aaai.org/ocs/index.php/AAAI/AAAI16/paper/view/12300}.

	\bibitem[Chen et~al.(2010)Chen, Lai, Parkes, and Procaccia]{CLPP10a}
	Y.~Chen, J.~K. Lai, D.~C. Parkes, and A.~D. Procaccia.
	\newblock Truth, justice, and cake cutting.
	\newblock In \emph{Proceedings of the 24th AAAI Conference on Artificial
	  Intelligence (AAAI)}, pages 756--761. AAAI Press, 2010.

	\bibitem[Filos-Ratsikas et~al.(2014)Filos-Ratsikas, Frederiksen, and
	  Zhang]{RKFZ14a}
	A.~Filos-Ratsikas, S.~K.~S. Frederiksen, and J.~Zhang.
	\newblock Social welfare in one-sided matchings: Random priority and beyond.
	\newblock In \emph{Proceedings of the 7th International Symposium on
	  Algorithmic Game Theory (SAGT)}, pages 1--12, 2014.

	\bibitem[Halpern et~al.(2020)Halpern, Procaccia, Psomas, and Shah]{HPPS20a}
	D.~Halpern, A.~D. Procaccia, A.~Psomas, and N.~Shah.
	\newblock Fair division with binary valuations: One rule to rule them all.
	\newblock Technical report, 2020.

	\bibitem[Hylland and Zeckhauser(1979)]{HyZe79a}
	A.~Hylland and R.~Zeckhauser.
	\newblock The efficient allocation of individuals to positions.
	\newblock \emph{The Journal of Political Economy}, 87\penalty0 (2):\penalty0
	  293--314, 1979.

	\bibitem[Katta and Sethuraman(2006)]{KaSe06a}
	A.-K. Katta and J.~Sethuraman.
	\newblock A solution to the random assignment problem on the full preference
	  domain.
	\newblock \emph{Journal of Economic Theory}, 131\penalty0 (1):\penalty0
	  231--250, 2006.

	\bibitem[Kavitha et~al.(2011)Kavitha, Mestre, and Nasre]{KMN11a}
	T.~Kavitha, J.~Mestre, and M.~Nasre.
	\newblock Popular mixed matchings.
	\newblock \emph{Theoretical Computer Science}, 412\penalty0 (24):\penalty0
	  2679--2690, 2011.

	\bibitem[Moulin(2003)]{Moul03a}
	H.~Moulin.
	\newblock \emph{Fair Division and Collective Welfare}.
	\newblock The MIT Press, 2003.

	\bibitem[Sethuraman(2010)]{Seth10a}
	J.~Sethuraman.
	\newblock Mechanism design for house allocation problems: A short introduction.
	\newblock \emph{OPTIMA: Mathematical Programming Society Newsletter}, pages
	  2--8, 2010.

	\bibitem[Varian(1974)]{Vari74a}
	H.~R. Varian.
	\newblock Equity, envy, and efficiency.
	\newblock \emph{Journal of Economic Theory}, 9:\penalty0 63--91, 1974.

	\bibitem[Vazirani(2007)]{Vazi07a}
	V.~V. Vazirani.
	\newblock Combinatorial algorithms for market equilibria.
	\newblock In N.~Nisan, T.~Roughgarden, {\'E}.~Tardos, and V.~Vazirani, editors,
	  \emph{Algorithmic Game Theory}, chapter~5, pages 103--134. Cambridge
	  University Press, 2007.

	\bibitem[Vazirani and Yannakakis(2020)]{VaYa20a}
	V.~V. Vazirani and M.~Yannakakis.
	\newblock {Computational Complexity of the Hylland-Zeckhauser Scheme for
	  One-Sided Matching Markets}.
	\newblock \emph{CoRR}, abs/2004.01348, 2020.
	\newblock URL \url{https://arxiv.org/abs/2004.01348}.

	\end{thebibliography}

\end{document}